\documentclass[11pt,letterpaper,english]{article}

\usepackage{amsmath,amssymb,amsthm}
\usepackage[ruled,vlined]{algorithm2e}

\usepackage{pgf, tikz}
\usetikzlibrary{arrows}
\usepackage{xspace, units}
\usepackage{typearea}

\usepackage{lmodern}
\usepackage[T1]{fontenc}
\usepackage{textcomp}

\typearea{16}

\usepackage{color}
\definecolor{Darkblue}{rgb}{0,0,0.4}
\definecolor{Brown}{cmyk}{0,0.81,1.,0.60}
\definecolor{Purple}{cmyk}{0.45,0.86,0,0}
\usepackage[breaklinks]{hyperref}
\hypersetup{colorlinks=true,
            citebordercolor={.6 .6 .6},linkbordercolor={.6 .6 .6},
            citecolor=black,urlcolor=Darkblue,linkcolor=black,
            pagecolor=Darkblue,
            pdfpagemode=UseThumbs,pdfstartview=FitH}
\newcommand{\lref}[2][]{\hyperref[#2]{#1~\ref*{#2}}}

\newtheorem{definition}{Definition}
\newtheorem{theorem}{Theorem}

\newtheorem{lemma}[theorem]{Lemma}
\newtheorem{corollary}[theorem]{Corollary}

\newtheorem{claim}[theorem]{Claim}
\newtheorem{observation}[theorem]{Observation}

\newcommand{\NN}{\ensuremath{\mathbb{N}}}

\newcommand{\R}{{\mathcal R}}

\renewcommand{\Pr}[1]{\mbox{\rm\bf Pr}\left[#1\right]}
\newcommand{\Ex}[1]{\mbox{\rm\bf E}\left[#1\right]}

\newcommand{\e}{\mathrm{e}}
\newcommand{\noise}{\nu}
\DeclareMathOperator{\avg}{avg}

\title{Dynamic Packet Scheduling in Wireless Networks}
\author{%
  Thomas Kesselheim%
  \thanks{Department of Computer Science, RWTH Aachen University, {\tt
      kesselheim@cs.rwth-aachen.de}. Supported by UMIC
    Research Centre at RWTH Aachen University.}  
}

\begin{document}

\maketitle
\begin{abstract}
We consider protocols that serve communication requests arising over time in a wireless network that is subject to interference. Unlike previous approaches, we take the geometry of the network and power control into account, both allowing to increase the network's performance significantly. 
We introduce a stochastic and an adversarial model to bound the packet injection. Although taken as the primary motivation, this approach is not only suitable for models based on the signal-to-interference-plus-noise ratio (SINR). It also covers virtually all other common interference models, for example the multiple-access channel, the radio-network model, the protocol model, and distance-2 matching. Packet-routing networks allowing each edge or each node to transmit or receive one packet at a time can be modeled as well.

Starting from algorithms for the respective scheduling problem with static transmission requests, we build distributed stable protocols. This is more involved than in previous, similar approaches because the algorithms we consider do not necessarily scale linearly when scaling the input instance. We can guarantee a throughput that is as large as the one of the original static algorithm. In particular, for SINR models the competitive ratios of the protocol in comparison to optimal ones in the respective model are between constant and $O(\log^2 m)$ for a network of size $m$.
\end{abstract}

\thispagestyle{empty}
\setcounter{page}0
\clearpage
\section{Introduction}
In order to exploit the full potential of wireless communication, it is crucial to suitably deal with the effect of interference, being one of the main limits of a wireless network's performance. Simultaneous transmissions may collide but only if they are not far enough apart. So in order to utilize the available time as efficient as possible, parallel communication has to be going on in spite of interference. In recent time, the resulting algorithmic problems have attained much interest, particularly the ones in the SINR model. Here, interference constraints are modeled much more realistically than in conventional models derived from graph theory. The SINR model takes accumulation of interference into account and allows to consider the effects of power control. That is, for each transmission an individual power can be selected. This technology has been shown to have a giant impact on the network's performance in theory as well as in practice. So far, algorithmic studies in the SINR model have mainly considered problems of the following flavor. Given a set of $n$ transmission requests, compute a schedule of minimum length such that all transmissions can be carried out successfully, possibly with the freedom of selecting the transmission powers. Although algorithmically challenging, this is a very limited view as it does not take into consideration that transmission requests actually arise over time.

For communication requests arriving in a network by adversarial injection or a stochastic process over time, algorithmic research has mainly considered two scenarios. In \emph{packet-routing networks} the focus lies on multi-hop communication in a wireline network. That is, packets have to use intermediate nodes until reaching their target node. The restriction is that on each communication link only a single packet may be transmitted in a time slot. In scheduling problems on a \emph{multiple-access channel}, a number of users have to share a channel but only one user can successfully transmit over the channel at a time. Although both approaches have also been applied in the context of wireless networks, they do not account for the geometry of the network. Packet routing networks neglect all effects of interference between communication links. In contrast, the multiple-access-channel model overestimates interference as it does not take the locality of interference into consideration.

In this paper, we aim at bridging the gap between these different settings. We consider a general model for dynamic packet injection that allows to take the  aspects of interference into consideration. For example, this includes the mentioned advantages of the SINR model such as the spatial separation of transmission but also different transmission powers and the fact if transmission powers are fixed or can be chosen by the protocol. In order to cover these different variants, our approach is quite general. Although the SINR model was the primary motivation, this has the interesting consequence that virtually all interference models are covered, such as the multiple-access channel, the radio-network model, the protocol model, and distance-2 matching. Furthermore packet-routing networks allowing that each edge or each node transmits or receives one packet at a time can be modeled as well.

We study \emph{stable} scheduling protocols. That is, the expected time for each packet from injection until delivery (latency) is bounded. Our objective is to build protocols of maximal throughput that guarantee stability. In order to express this performance, we say that a protocol is $\gamma$-\emph{competitive} if the following holds. Assuming that there is some way an optimal protocol could serve all arising transmission requests, our protocol would be able to do so as well if time was stretched by the factor $\gamma$. Technically, we consider transmission requests arising from stochastic and adversarial injection. In the stochastic model, the injection by a finite number of independent users has to be a convex combination of feasible sets, scaled by factor $\gamma$. In addition to that, we adapt the popular model of a window adversary. During each interval of length $w$, the adversary may only inject packets that could be served in that time, scaled by $\gamma$. 

We make use of the results for scheduling static transmission requests by giving a black-box transformation to the dynamic model. For a number of algorithms, the exact throughput bound can be transferred to the dynamic case. Using comparisons to the optimal static schedule length, this yields competitive ratios that are as good as the approximation factor of the static algorithms. 

\subsection{Our Contribution}
We introduce a model for adversarial and stochastic packet injection in a wireless network that is suitable for SINR models. Like in a packet-routing network, injected packets may have to use intermediate nodes in order to reach the final destination. For this purpose, the network nodes correspond to vertices of a graph that are connected by an edge if a transmission can take place between them. Due to the diversity in assumptions, we model the aspects of interference in a generic and abstract way. We choose a suitable matrix $W$ quantifying the relative amount of interference of one edge on another one. Based on this matrix, we define an adversarial and a stochastic injection model, limiting the average amount of packets injected per time slot by the injection rate $\rho$ as follows. If $F(e)$ is the average number of packets using edge $e$ then each entry of the vector $W \cdot F$ has to be bounded by $\rho$. 

The definition is motivated by \emph{linear interference measures} for the SINR model. Consider a static single-hop instance in which $n$ packets have to be transmitted from their sender to the respective receiver, the interference measure defined by the matrix $W$ is given as $I = \max_{e \in E} \sum_{e' \in E} W_{e, e'} R(e')$. Here, $R(e)$ denotes the number of packets that have to be transmitted via the edge $e$.  Examples of existing static scheduling algorithms generate schedules of length $O(I \cdot \log n)$ or $O(I + \log^2 n)$ with high probability for the respective interference measure $I$.

In related work \cite{Goldberg2000,Rabani1996,Scheideler2000} typically a protocol for dynamic injection is built by repeatedly running a static algorithm for a suitably long time. In our case this does not have the desired effect in general. For example, when scaling the number of communication requests per edge, an $O(I \cdot \log n)$ schedule length increases super-linearly since both $I$ and $n$ increase. Thus, having more packets the throughput decreases. In order to deal with this problem we show in the first step how to transform these algorithms to ones that are suitable for dense instances. We exploit the fact that there are only $m$ possible communication links. This allows us to improve the scaling behavior of an algorithm computing schedules of length $O(f(n) \cdot I)$ with high probability to $O(f(m) \cdot I + g(m, n))$, where $f(m)$ only depends on the network size and $g(m, n)$ grows sub-linearily in $n$. 

The algorithms resulting from the first transformation are suitable to be used in the dynamic scenario. Here, we divide time into sufficiently long time frames. In each of them, the static algorithm is executed with the intention that each injected packet is transmitted via one hop in each time frame. However, packets may fail due to too many injected packets or collisions in the algorithm. These packets are treated by separate executions of the algorithm. This protocol is shown to be stable for injection rates corresponding to the throughput of the respective static algorithm. As a result, we obtain the static approximation factor as the competitive ratio. In particular, for the SINR model we achieve competitive ratios between constants and $O(\log^2 m)$. The expected latency of a packet is also shown to be poly-logarithmic in the size of the network.

Depending on the properties of the algorithm, we obtain a distributed protocol. In order to run the transformation, the network nodes only need the knowledge of a global clock, and the size of the network, the injection rate and (for the case of adversarial injection) the window size. It is reasonable that this information is available to each network node as it is static, that is, it does not depend on the packet injections and can be set at the deployment time. Furthermore, we show that being aware of a global clock is inevitable. When assuming only local clocks, no protocol for the SINR model with uniform transmission powers can be $m / 2 \ln m$-competitive. Achieving $O(m)$-competitiveness is trivial by falling back to the multiple-access channel model.

\subsection{Related Work}
The analysis of stochastically arriving transmission requests in a wireless network has first been considered in the context of ALOHA \cite{Abramson1970}. Here, a multiple-access channel is considered, that is only one transmission request can be served at a time. Over the years, this work has been continued under a large number of different assumptions, e.g. if there are finitely or infinitely many users or how much feedback the transmitters get, see e.g. \cite{Hastad1996,Raghavan1998,Goldberg2000,Goldberg2000a} or \cite{Chlebus2001} for an overview. A related model and problem class to be mentioned here is broadcast in radio networks \cite{BarYehuda1992,Clementi2001,Demaine2006}. Here, a simple natural extension of the multiple-access channel is considered by assuming that not all nodes can hear each other but messages might have to take intermediate nodes. A transmission is successfully received by a node if exactly one of its neighbors makes a transmission attempt.

A different approach for dynamic scheduling in wireless networks has been considered by Tassiulas and Ephremides~\cite{Tassiulas1992}. They consider a network with arbitrary interference constraints, where in each round transmission requests arise by an independently, identically distributed process. Tassiulas and Ephremides prove optimality of a protocol that selects in each round a maximum weight set of communication links. The protocol is optimal because it is stable for any injection for which there is some stable protocol. However, this protocol is neither distributed nor can it be computed in polynomial time in general. Viewed from this perspective, we show how to approximate this optimal protocol.

The probably most popular approach to bound adversarial packet injection was presented by Borodin et al.~\cite{Borodin2001} and refined by Andrews et al.~\cite{Andrews1996}. The general idea is that there is some window size $w$. The adversary is $(\lambda, w)$-bounded if during any interval of $w$ time steps for any edge $e \in E$ at most $\lambda \cdot w$ packets are injected having the edge $e$ on their path. Andrews et al.~show that very simple local policies such as \emph{shortest-in-system} (SIS) guarantee that for each $\lambda < 1$ the number of undelivered packets in the system is bounded at any time. The protocol by Aiello et al.~\cite{Aiello1998} achieves essentially the same result but does not have to know the routing paths. It only suffices that there are paths (only known to the adversary) that make the adversary $(\lambda, w)$-bounded.

The model of a $(\lambda, w)$-bounded adversary has also been applied to the multiple access channel~\cite{Bender2005,Chlebus2006,Chlebus2007}.  The idea is that in each time interval of length $w$ at most $\lambda \cdot w$ packets can arrive. Chlebus et al.~\cite{Chlebus2006} show that quite simple deterministic protocols are stable for all $\lambda < 1$, whereas stability for $\lambda = 1$ is impossible for distributed protocols. Further adaptations of the window adversary also consider wireless networks \cite{Andrews2004,Cholvi2010}. However, in these cases interference is again completely neglected. To the best of our knowledge adversarial injection taking locality of interference into account has not yet been considered.

While commonly the only criterion is bounded delay, Rabani and Tardos~\cite{Rabani1996}, and Scheideler and V\"ocking~\cite{Scheideler2000} show how to achieve small delays by transforming static packet-routing algorithms. The second part of our transformation is inspired by the one of Scheideler and V\"ocking and structurally similar. However, in order to achieve stability they use SIS as a fallback solution, which is known to yield stability. This is not possible in our case since no stable protocols have been known up to now. Furthermore, their analysis and way to cope with dependencies is complex and tailored to the packet-routing case. For this reason, our analysis does not have much in common with the one by Scheideler and V\"ocking.

Recently, algorithm research has started considering scheduling wireless transmissions in the SINR model. As already mentioned, this model accounts for accumulation effects of interference and the possibility of selecting transmission powers. These different transmission powers offer a new degree of freedom to the problem depending on if the powers are specified as part of the input or if they can be set by the algorithm. For the first problem of dealing with fixed transmission powers, e.g. uniform powers, a number of algorithms have been proposed, centralized \cite{WattenhoferINFOCOM09,Halldorsson2009} and distributed \cite{Fanghaenel2009,Kesselheim2010,Halldorsson2011a} ones. For the problem in which transmission powers can be set by the algorithm one approach is to set transmission powers obliviously, depending on the distance between the sender and the receiver \cite{Fanghaenel2009a,Halldorsson2009a,Halldorsson2011}. While this approach can achieve only trivial approximation factors in terms of $n$, there is also a centralized algorithm achieving an $O(\log n)$ approximation guarantee \cite{Kesselheim2011}.

Further studies in the context of SINR approximation algorithms include distributed capacity maximization \cite{Andrews2009,Dinitz2010,Asgeirsson2011} and online capacity maximization \cite{OnlineSPAA}. When considering conflict graphs with a small inductive independence number \cite{Hoefer2011,Hoefer2011a}, a similar abstraction as in this paper is used. However, we aim at building distributed protocols in this paper whereas the approach in \cite{Hoefer2011} is rather appropriate for centralized, LP-based approximations.
\section{Formal Definition of the Network Model}
\label{sec:networkmodel}
We assume the wireless network to be modeled as a directed graph $G = (V, E)$. The vertex set $V$ corresponds to the set of network nodes. The set $E$ indicates the set of possible communication links between two nodes. As the graph is not necessarily complete, we assume that packets might need to be transmitted via intermediate nodes before reaching their final destination. These paths are fixed for each packet, e.g., by routing tables. They may, in principle, visit nodes multiple times. They are only restricted to have length at most $D$. We will use $m := \max\{ \lvert E \rvert, D\}$ as the significant network size.  

Via each communication link at most one packet may be transmitted per time step. Furthermore, transmissions on different links are also subject to interference. Generalizing multiple variants of the SINR model but also packet-routing networks, the multiple access channel and a broad class of wireless interference models (see Sections~\ref{sec:sinr} and \ref{sec:furtherapplications}), we employ a \emph{linear interference measure}. This is, there is some matrix $W$ to express the (relative) impact that a transmission on one link has to a transmission on another one. It is chosen later on based on the geometry and the interference assumptions. More precisely, for two edges $e$ and $e'$, the quantity $W_{e, e'} \in [0, 1]$ indicates, how much a transmission on $e$ is interfered by a transmission on $e'$. We assume that $W_{e, e} = 1$ for all $e \in E$. Given a set of paths let $R(e)$ denote the number of paths including edge $e$ somewhere. The interference measure induced by the vector $R$ is now given by $I := \lVert W \cdot R \rVert_\infty = \max_{e} \sum_{e'} W_{e, e'} \cdot R(e')$.

The case of packet routing networks can be captured by setting $W$ to the identity matrix and receive the congestion as the interference measure. For the multiple-access channel, we can set $W$ to the matrix whose entries are all $1$. In this case $I$ is simply the total number of packets to be delivered.

\subsection{Injection Models}
For the communication requests arising over time, we adapt two famous models: time-independent, finite-user stochastic injection and injection by a window adversary. In either case, the injected packets are assumed to have a fixed path through the network. We bound the average interference measure of all communication requests injected per time slot. If $F(e)$ is the average number of packets that have to be transmitted via edge $e$, the injection rate $\lambda$ is the largest component of the vector $W \cdot F$. 

For the stochastic model, we take the following assumptions. We assume that there is a finite number of packet generators each of which injects at most one packet per time slot at random. The probability distribution is identical in each time slot and independent among different generators or different time slots. Formally, let $X_{g, P}^{t}$ be $1$ if generator $g$ injects in time slot $t$ a packet that shall be routed along path $P$. We assume these random variables to have the following three properties. (a) The injection in each time step is identically distributed. That is for any pair $t_1$ and $t_2$ the random variables $X_{g, P}^{t_1}$ and $X_{g, P}^{t_2}$ have to be identically distributed for all $g$ and $P$. (b) The injection of different generators and in different time slots is independent. Formally, we require independence of any subset of random variables $X_{g, P}^{t}$ in which no pair shares both the same $t$ and the same $g$. (c) Each generator only injects a single packet per time slot. That is for any fixed $t$, and $g$ only one of the $X_{g, P}^{t}$ can be 1. We require each component of $W \cdot F$ to be bounded by $\lambda$, where $F(e) = \sum_g \sum_{P: e \in P} \Ex{X^{t}_{g, P}}$.

Furthermore, we consider a $(w, \lambda)$-bounded adversary for an arbitrary $w \in \NN$. That is, considering an arbitrary interval of $w$ time slots, we require that the interference measure induced by all links of the respective paths is at most $w \cdot \lambda$. Formally, let $R(e)$ be the number of packets including edge $e$ on the path injected during that interval. Then each component of the vector $W \cdot R$ is bounded by $w \cdot \lambda$.
\section{Static Algorithms for Large Packet Numbers}
\label{sec:largepacketnumbers}
Existing approaches to use static scheduling algorithms in a dynamic environment \cite{Goldberg2000,Rabani1996,Scheideler2000} all share the idea of running the algorithm repeatedly for a suitably long time. As in these cases the schedule length grows linearly, this does not decrease the throughput and at the same time failures are less likely. In our case, however, the situation is different. Consider for example an algorithm that computes a schedule of length $O(I \cdot \log n)$ for $n$ packets with high probability. Then doubling all packets does not only double the number of time slots used as both $I$ and $n$ are doubled. Our solution to this problem is to exploit that there are only $m$ different links that can be used for transmissions. Starting e.g. from an $O(I \cdot \log n)$ algorithm, our transformation yields an $O(I \cdot \log m + \log n \cdot \log^2 m + \log^2 n \cdot \log m)$-algorithm. That is, for sufficiently many transmission requests, the schedule length becomes linear in $I$. 

More precisely, we assume that there is some algorithm $\mathcal{A}(I, n)$ that generates a schedule of length $f(n) \cdot I$ with probability $1 - \nicefrac{1}{n}$ if the interference measure is at most $I$ and the number of packets is at most $n$. Algorithm~\ref{alg:largepacketnumbers} runs $\mathcal{A}$ repeatedly on randomly selected subsets of the communication requests. Each edge randomly selects delay values for all waiting packets. The algorithm is then executed on all packets having received the same delay, assuming the interference measure of these packets is at most $\chi = 6(\ln m + 9)$. At some point not many packets remain and we can go back to the original algorithm.

\begin{algorithm}
\label{alg:largepacketnumbers} 
\caption{Transformation to get a schedule length that is independent of the number of packets}
\For{$i:=1$ to $\xi = \left\lceil \log(I / 2 \phi \chi \log n) \right\rceil$}{
Assign each remaining packet a delay of at most $\left\lceil \nicefrac{2^{-i + 1} \cdot I}{\chi} \right\rceil$\;
Execute $\mathcal{A}(\chi, m \cdot \chi)$ for $f(m \chi) \cdot \chi$ steps on all packets that received the same delay\;
}
\For{$i:=1$ to $\lceil\phi\rceil + 1$}{
Execute $\mathcal{A}(2 \phi \chi \log n, n)$ on the remaining packets\;
}
\end{algorithm}

In the analysis, we show that with high probability the interference measure induced by the remaining transmission requests reduces by a factor of two in each iteration of the for loop. Thus, after $\xi$ iterations, the interference measure has become as small as $O(\log n \cdot \log m)$ and for this reason the original algorithm can schedule all remaining packets with it in $O(f(n) \cdot \log n \cdot \log m)$ steps. 

\begin{theorem}
\label{theorem:largepacketnumbers} 
If $\mathcal{A}(I, n)$ uses at most $f(n) I$ steps with probability at least $1 - \nicefrac{1}{n}$, then for each constant $\phi \geq 0$ Algorithm~\ref{alg:largepacketnumbers} uses at most $2 \cdot f(m \chi) \cdot I + O(\log n \cdot f(m \chi) + f(n) \cdot \log n \cdot \log m )$ with probability at least $1 - \nicefrac{1}{n^\phi}$.
\end{theorem}

\begin{proof}
The number of time slots Algorithm~\ref{alg:largepacketnumbers} is simply the sum of the numbers of time slots used by all iterations of $\mathcal{A}$  
\begin{align*}
& \sum_{i=1}^\xi \left\lceil \frac{2^{-i + 1} \cdot I}{\chi} \right\rceil \cdot f(m \chi) \cdot \chi + (\lceil\phi\rceil + 1) \cdot f(n) \cdot 2 \phi \chi \log n \\ 
\leq\; & 2 \cdot f(m \chi) \cdot I + \xi \cdot f(m \chi) \cdot \chi + (\lceil\phi\rceil + 1) \cdot f(n) \cdot 2 \phi \chi \log n \\ 
=\; & 2 \cdot f(m \chi) \cdot I + O(\log n \cdot f(m \chi) + f(n) \log n \cdot \log m ) \enspace.
\end{align*}

In order to bound the success probability, let $R^{(i)}(e)$ be the number of remaining packets for edge $e$ after the $i$th iteration of the for loop, $R^{(0)} = R$. We claim that with probability at least $1 - \nicefrac{1}{n^\phi}$ all of the following events occur:
\begin{itemize}
  \item For all $i \in [\xi]$, we have $\lVert W \cdot R^{(i)} \rVert_\infty \leq 2^{-i} \cdot I$
  \item All remaining packets are successfully transmitted in one of the last $\phi \log n$ executions of $\mathcal{A}$. 
\end{itemize} 

In order to bound the probability of a failure, we consider the first event that does not occur. That is, let us assume that for $i$, we have $\lVert W \cdot R^{(i)} \rVert_\infty \leq 2^{-i} \cdot I$. Given this event, we now bound the probability that $\lVert W \cdot R^{(i+1)} \rVert_\infty \leq 2^{-i-1} \cdot I$.

\begin{claim}
For all $i \in [\xi]$
\[
\Pr{\lVert W \cdot R^{(i)} \rVert_\infty > 2^{-i} \cdot I \mid \lVert W \cdot R^{(i - 1)} \rVert_\infty \leq 2^{-(i-1)} \cdot I} \leq \frac{2^\xi}{2^i} \cdot \frac{1}{4n^\phi} \enspace.
\]
\end{claim}

\begin{proof}
Let $i \in [\xi]$ be fixed. The packets listed in $R^{(i)}$ have not been transmitted properly up to the $i$th iteration of the algorithm. In the $i$th iteration each packet from $R^{(i-1)}$ is assigned a delay value uniformly at random from the set $[\psi]$, where $\psi = \left\lceil \nicefrac{2^{-i + 1} \cdot I}{\chi} \right\rceil$.

Let $R^{(i - 1)}_j(e)$ be the number of packets having been assigned delay $j$ in the $i$th iteration.

A packet might not be successfully transmitted in the $i$th iteration for two reasons. Either we have $\lVert W \cdot R_j^{(i - 1)} \rVert_\infty > \chi$ or in spite of the fact that the interference measure was small enough the algorithm failed.

In the first case, we set $Y_j = \lVert W \cdot R_j^{(i - 1)} \rVert_\infty / \chi$, otherwise we set $Y_j = 0$. If in contrast, we had a failure of the algorithm, we set $Z_j = 1$ and otherwise $Z_j = 0$. These definitions yield $\lVert W \cdot R^{(i)} \rVert \leq \sum_{j = 1}^{\psi} Y_j \cdot \chi + \sum_{j = 1}^{\psi} Z_j \cdot \chi$.

Let us first consider the random variables $Y_j$. For fixed $j \in [\psi]$ and $e \in E$ the random variable $(W \cdot R_j^{(i+1)})(e)$ is the weighted sum of independent random variables because for each packet the delay is chosen independently. For this reason, we can apply a Chernoff bound on $(W \cdot R_j^{(i)})(e)$. For the expectation, we have
\[
\Ex{(W \cdot R_j^{(i-1)})(e)} \leq \frac{\lVert W \cdot R^{(i-1)} \rVert_\infty}{\psi} \leq \frac{\chi}{2} \enspace.
\]
Chernoff and union bound now yield that for all $\delta \geq 1$, we have
\[
\Pr{\lVert W \cdot R_j^{(i-1)} \rVert_\infty \geq (1 + \delta) \cdot \frac{\chi}{2}} \leq \sum_{e \in E} \Pr{(W \cdot R_j^{(i-1)})(e) \geq (1 + \delta) \cdot \frac{\chi}{2}} \leq m \cdot \exp\left( - \frac{\delta \chi}{6} \right) \enspace. 
\]
In terms of $Y_j$ this is for all $\alpha \geq 1$ 
\begin{align*}
\Pr{Y_j \geq \alpha} & \leq m \cdot \exp\left( - \frac{(2 \alpha - 1) \chi}{6} \right) \leq m \cdot \exp\left( - \frac{\chi}{6} \alpha \right) \\
& \leq m \cdot \exp\left( - (\ln m + 8 + \ln 2) \alpha \right) \leq \exp(- 8 \alpha) \cdot 2^{-\alpha} \enspace,
\end{align*}
which yields
\[
\exp(4\alpha + 4) \cdot \Pr{Y_j \geq \alpha} \leq \frac{1}{2^\alpha} \enspace.
\]
Thus, we can conclude for $\Ex{\e^{4 \cdot Y_j}}$
\[
\Ex{\e^{4 \cdot Y_j}} \leq 1 + \sum_{a = 2}^\infty \exp(4a) \cdot \Pr{Y_j \geq a-1} = 1 + \sum_{\alpha=1}^\infty \exp(4\alpha + 4) \cdot \Pr{Y_j \geq \alpha} \leq 1 + \sum_{\alpha=1}^\infty \frac{1}{2^\alpha} = 2 \enspace.
\]

The random variable $Y_j$ are not independent. Nevertheless, we have $\Ex{\prod_{j=1}^{\psi} \e^{4 \cdot Y_j}} \leq \prod_{j=1}^{\psi} \Ex{\e^{4 \cdot Y_j}}$. This is due to the fact that $\e^{Y_1}$ is non-decreasing in $Y_1$ and $\Ex{\prod_{j \neq 1} \e^{Y_j} \mid Y_1}$ is non-increasing in $Y_1$ because larger values of $Y_1$ only reduce the probability that there are many packets with delay $2, \ldots, \psi$. The FKG inequality now yields $\Ex{\e^{Y_1} \cdot \Ex{\prod_{j \neq 1} \e^{Y_j} \mid Y_1}} \leq \Ex{\e^{Y_1}} \cdot \Ex{\Ex{\prod_{j \neq 1} \e^{Y_j} \mid Y_1}} = \Ex{\e^{Y_1}} \cdot \Ex{\prod_{j \neq 1} \e^{Y_j}}$. Applying this argument repeatedly yields the claim.

For this reason, we get
\[
\Pr{\sum_{j=1}^{\psi} Y_j \geq \frac{\psi}{4}} = \Pr{\e^{4 \sum_{j=1}^{\psi} Y_j} \geq \e^{\psi}} \leq \e^{-\psi} \prod_{j=1}^{\psi} \Ex{\e^{4 \cdot Y_j}} \leq \e^{-\psi} \cdot 2^{\psi} \leq 2^{- \frac{\psi}{4}} \enspace.
\]

The random variables $Z_j$ are independent $0$/$1$ random variables. For each of them, we have $\Ex{Z_j} \leq \frac{1}{8 \e}$, and therefore for the sum $\Ex{\sum_{j=1}^{\psi} Z_j} \leq \frac{\psi}{8 \e}$. Applying a Chernoff bound yields
\[
\Pr{\sum_{j = 1}^{\psi} Z_j \geq \frac{\psi}{4}} \leq 2^{- \frac{\psi}{4}} \enspace.
\]

Thus, we get
\[
\Pr{\lVert W \cdot R^{(i)} \rVert_\infty > 2^{-i} \cdot I \mid \lVert W \cdot R^{(i - 1)} \rVert_\infty \leq 2^{-(i-1)} \cdot I} \leq \Pr{\sum_{j=1}^{\psi} Z_j + \sum_{j=1}^{\psi} Y_j \geq \frac{\psi}{2}} \leq 2 \cdot 2^{-\frac{\psi}{4}} \enspace.
\]
Furthermore, we have
\[
\psi = \left\lceil \frac{2^{-i + 1} \cdot I}{\chi} \right\rceil \geq 4 \left( \xi - i + 2 \phi \log n \right) \enspace.
\]
For this reason
\[
2^{-\frac{\psi}{4}} \leq 2^{\xi - i + 2 \phi \log n} \leq \frac{2^\xi}{2^i} \cdot \frac{1}{4n^\phi} \enspace.
\]
The completes the proof of the claim.
\end{proof}

Having shown this bound for each iteration, we can now take the sum over all $i \in [\xi]$ to get
\[
\sum_{i=1}^{\xi} \Pr{\lVert W \cdot R^{(i)} \rVert_\infty > 2^{-i} \cdot I \mid \lVert W \cdot R^{(i - 1)} \rVert_\infty \leq 2^{-(i-1)} \cdot I} \leq \sum_{i=1}^{\xi} \frac{2^\xi}{2^i} \cdot \frac{1}{4n^\phi} \leq 2 \cdot \frac{1}{4n^\phi} \enspace.
\]

Now let us consider the last $\lceil \phi \rceil + 1$ executions of $\mathcal{A}$. Provided that $\lVert W \cdot R^{(\xi)} \rVert_\infty \leq 2^{-\xi} \cdot I \leq 2 \phi \chi \log n$, by our assumption the probability that not all packets are successfully transmitted in one execution is at most $\nicefrac{1}{n}$. Having $\lceil \phi \rceil + 1$ independent repeats, this failure probability reduces to $\nicefrac{1}{2 n^\phi}$. Taking another union bound, this shows that the combined failure probability is at most $\nicefrac{1}{n^\phi}$.
\end{proof}

\section{Dynamic Scheduling Protocol for Stochastic Injection}
We are now prepared to transform the static algorithm into a protocol for dynamic packet injection. In this section, we consider the stochastic injection. In the next section, the results are transferred to the adversarial injection model. The assumption we make is that there is some algorithm $\mathcal{A}(I, n)$. Given at most $n$ communication requests of interference measure at most $I$, it computes a schedule of length $f(m) \cdot I + g(m, n)$ with probability at least $1 - \nicefrac{1}{2 n^4}$. Here, $f$ is a function independent of $n$, and $g$ is a function growing sublinearily in $n$. Given this algorithm, we build a stable protocol for each injection rate $\lambda < \frac{1}{f(m)}$.

Let $\lambda = (1 - \varepsilon) / f(m)$. Without loss of generality, we assume that $\varepsilon \leq \nicefrac{1}{2}$. We divide time into frames of length $T$. We require that $T \geq \frac{100 f(m)}{\varepsilon^3} + \frac{48 f(m) \ln m}{\varepsilon^2}$ and furthermore that $T \geq \frac{4 f(m)}{\varepsilon^2} g(m, \frac{m}{f(m)} \cdot T)$. The latter condition is fulfilled for sufficiently large $T$ because $g(m, n)$ grows sublinearily in $n$. For example, if $f(m) = O(\log m)$ and $g(m, n) = O(\log n \cdot \log^2 m + \log^2 n \cdot \log m)$, as derived in the previous section, it suffices to have $T = O(\frac{\log^4 m}{\varepsilon^2})$. Furthermore, we set $J = (1+\varepsilon) \cdot \lambda \cdot T$.

Each time frame of length $T$ itself consists of two phases. The first phase is dedicated to \emph{unfailed packets}. Each packet is intended to make one hop towards its final destination in each phase. In order to achieve this goal, after injection a packet waits for the next time frame to begin. Here, $\mathcal{A}(J, m \cdot J)$ is executed for $T' = f(m) \cdot J + g(m, m \cdot J)$ time slots on the set containing the respective next hop on the path of each packet that has not failed so far. In this execution packets can fail to reach their next hop destination. If this happens, a packet is referred to as \emph{failed} and will from now on be only scheduled for transmission in the second phase, the \emph{clean-up phase}. The clean-up phase consists of the remaining $T - T'$ time slots of the time frame. Here, the algorithm is executed  another time but only on the following set of packets. Each edge $e$ with a non-empty buffer of failed packets performs a random experiment. With probability $\nicefrac{1}{m}$ it selects the failed packet from its buffer whose failure is longest ago. With the remaining probability no packet from the buffer of failed packets on this edge is selected in this round. On the selected packets, we execute $\mathcal{A}(1, m \cdot J)$ for $f(m) \cdot 1 + g(m, m \cdot J)$ time steps. If $T$ fulfills the bounds mentioned before, these are at most $T - T'$ steps. So both phases fit into a time frame.

In order to prove stability, we have to consider the failed packets. Each packet that does not fail will reach its final destination after at most $D$ time frames. The central question is therefore whether the clean-up phases are able to keep the buffers of failed packets small. In the following two sections, we show that both queue lengths and packet latency are bounded in expectation, proving the protocol to be stable. 

\subsection{Queue Lengths}
In order to show the stability of the protocol, we show in this section that in expectation queue lengths are bounded. As previously stated, it suffices to bound the lengths of buffers for failed packets. Packets that do not fail spend at most $D + 1$ time frames in the system. Having a bounded (expected) number of packets injected per time step, they do not have to be considered anymore.

\begin{theorem}
The expected queue lengths (i.e. number of undelivered packets) are bounded at any time.
\end{theorem}

To prove the theorem, we consider as a potential function $\Phi$ the sum of the numbers of remaining hops all failed packets have to cross. In a clean-up phase this quantity reduces if a transmission is successfully carried out. Obviously, this potential function is an upper bound on the summed buffer sizes as well.

First, we bound the increase of the potential function in a time frame. This is due to colliding packets. The increase may depend on the previous value of the potential function. For example, if all packets collided in the previous time frames, collisions are less likely. Fortunately, we can use the following pessimistic assumption. The probability of a collision is maximal (and therefore the potential increase) if no packets have collided before. Therefore we will assume for the bound on the potential increase that all injected packets have reached the current time frame without failing. This may yield that we account for failed packets multiple times: We add its contribution to the potential function in each time frame it would fail. However, this assumption allows us to treat the potential like a Markov chain.

\begin{lemma}
\label{lemma:potentialincrease}
For each $i \in \NN$ the probability that the potential increases by at least $i \cdot m^2 J + 1$ is at most $\left( m J \right)^{- 4 - i}$.
\end{lemma}

\begin{proof}
Let $I$ be the interference measure of all transmission requests that were originally meant to be served in this phase. As we have $W_{e, e} = 1$ for all $e \in E$ and path lengths are at most $D$, the potential increase in case of a failure can be bounded by $D \cdot \lvert E \rvert \cdot I \leq m^2 I$. We now distinguish between the two cases $i=0$ and $i > 0$.

\begin{claim}
\label{claim:potentialincrease:chernoff}
For all $\delta > 0$, we have $\Pr{ I \geq (1 + \delta) \lambda T} \leq m \cdot \left( \frac{\e^\delta}{(1 + \delta)^{1 + \delta}} \right)^{\lambda T}$.
\end{claim}

\begin{proof}
We denote the packets that are injected in time step $t$ as follows. The variable $X^{g, e}_{t, d}$ indicates if generator $g$ injects a packet having edge $e$ as the $d$th hop in time step $t$. The $X^{g, e}_{t, d}$ are random variables having the following three properties. 
\begin{itemize}
  \item The injection in each time step is identically distributed. That is, for any pair $t_1$ and $t_2$ the random variables $X_{t_1, d}^{g, e}$ and $X_{t_2, d}^{g, e}$ have to be identically distributed for all $g$, $e$ and $d$.
 \item The injection of different generators and in different has to be independent. This can be formalized by stating independence for each subset $S$ of these random variables with the following property. If $X_{t, d}^{g, e}, X_{t', d'}^{g', e'} \in S$ and $X_{t, d}^{g, e} \neq X_{t', d'}^{g', e'}$, then we have $t \neq t'$ or $g \neq g'$.
 \item Each generator only injects a single packet per time slot. That is for any fixed $t$, $d$ and $g$ only one of the $X_{t, d}^{g, e}$ can be 1.
\end{itemize}

Let $R(e)$ be the number of packets attempting to be transmitted via $e$ in the time frame of consideration. Applying the above notation, we have $R(e) = \sum_g \sum_{j=1}^d \sum_{t=t_0 - (j+1)T - 1}^{t_0 - jT}  X_{t_0, j}^{g, e}$, where $t_0$ is the first time slot of the currently considered time frame. For this reason $(W \cdot R)(e)$ is a weighted sum of independent random variables. Its expectation is
\[
\Ex{(W \cdot R)(e)} = \Ex{\sum_{e' \in E} W_{e, e'} \sum_g \sum_{j=1}^D \sum_{t=t_0 - (j+1)T - 1}^{t_0 - jT} X_{t_0, j}^{g, e'}}
= \sum_{e' \in E} W_{e, e'} \sum_g \sum_{j=1}^D \sum_{t=t_0 - (j+1)T - 1}^{t_0 - jT}  \Ex{X_{t_0, j}^{g, e'}} \enspace.
\]
The injection is time-invariant. So this is equal to
\[
\sum_{e' \in E} W_{e, e'} \sum_g \sum_{j=1}^D T \Ex{X_{t_0, j}^{g, e'}} \leq \lambda T \enspace.
\]
Applying union bound and Chernoff bound, we get
\[
\Pr{ I \geq (1 + \delta) \lambda T} \leq \sum_{e \in E} \Pr{ (W \cdot R)(e) \geq (1 + \delta) \lambda T } \leq m \cdot \left( \frac{\e^\delta}{(1 + \delta)^{1 + \delta}} \right)^{\lambda T} \enspace. 
\]
\end{proof}

For $i = 0$, we have to bound the probability that a packet fails. As mentioned earlier there are two possible reasons for this event to occur. On the one hand, the network is overloaded in the time frame, that is $I > J$. In order to apply Claim~\ref{claim:potentialincrease:chernoff}, we use the fact that $T \geq \frac{100 f(m)}{\varepsilon^3} + \frac{48 f(m) \ln m}{\varepsilon^2}$ to get $\lambda T \geq \frac{50}{\varepsilon^3} + \frac{24 \ln m}{\varepsilon^2}$. This implies
\[
\lambda T \geq 2 A \ln\left(A m (1 + \varepsilon)\right) \enspace, \quad \text{where $A = \frac{5}{- \ln\left( \frac{\e^\varepsilon}{(1+\varepsilon)^{1+\varepsilon}} \right)}$} \enspace.
\]
As for all $x \geq 0$, we have $x \geq \frac{x}{2} + \ln(x)$, this yields
\[
\frac{\lambda T}{A} \geq \frac{\lambda T}{2A} + \ln\left(\frac{\lambda T}{A}\right) \geq \ln\left(A m (1 + \varepsilon)\right) + \ln\left(\frac{\lambda T}{A}\right) = \ln\left(m (1 + \varepsilon) \lambda T\right) = \ln\left(m J\right) \enspace.
\]
Putting this into the bound on $\Pr{I \geq J}$ obtained by Claim~\ref{claim:potentialincrease:chernoff}, we get
\[
\Pr{ I \geq J } = \Pr{ I \geq (1+\varepsilon) \lambda T} \leq m \left( \frac{e^\varepsilon}{(1+\varepsilon)^{1+\varepsilon}}\right)^{A \ln(m J)} = m \left( \frac{1}{m J} \right)^5 = \frac{1}{2 (m J)^4} \enspace.
\]

Still in the case $I \leq J$ packets may fail. This is due to the fact that internal randomization of the algorithm can result in failures. We required the algorithm to have at failure probability of at most $\nicefrac{1}{2 (m J)^4}$ in this case. Combining the two bounds, this shows the claim for the case $i = 0$.

Considering the case $i > 0$, we apply again the fact that $T \geq \frac{100 f(m)}{\varepsilon^3} + \frac{48 f(m) \ln m}{\varepsilon^2}$. This implies $T \geq \frac{100 + 30 \ln m}{\lambda}$. This is, we have
\[
\frac{\lambda T}{15} \geq \frac{\lambda T}{30} + \ln\left( \frac{\lambda T}{15} \right) \geq \ln(15 (1 + \varepsilon) m) + \ln\left( \frac{\lambda T}{15} \right) = \ln(m J) \enspace.
\]
Putting this into our bound by Claim~\ref{claim:potentialincrease:chernoff}, we get
\[
\Pr{ I \geq (1+i) J} \leq \Pr{ I \geq (1+i) \lambda T} \leq m \exp\left( - \frac{i \cdot \lambda T}{3} \right) \leq  m \left( \frac{1}{m J} \right)^{5 i} \leq  \left( \frac{1}{m J} \right)^{4 + i} \enspace. 
\]
\end{proof}

For the potential decrease in clean-up phases we use a very pessimistic but simple assumption. In the worst case all failed packets are in the same buffer. Even in this case, the potential decreases with probability at least $\nicefrac{1}{2 \e m}$.

\begin{lemma}
\label{lemma:potentialdecrease}
The probability that a non-zero potential decreases is at least $\nicefrac{1}{2 \e m}$.
\end{lemma}

\begin{proof}
Having non-zero potential, at least one buffer contains failed packets. For this reason, the probability that at least one packet is selected is at least $\nicefrac{1}{m}$. With probability at least $(1 - \nicefrac{1}{m})^{m-1} \geq \nicefrac{1}{\e}$ no other packet is selected. The success probability of the algorithm is at least $\nicefrac{1}{2}$. All events are independent.
\end{proof}

Combining these two bound we get the following facts on the Markov chain's drift, that is its expected change. The drift is finite for each state and in the case of non-zero potential it is negative. This already yields that the Markov chain is ergodic \cite{Pakes1969}. However, we can also bound the probability distribution quantitatively. 

\begin{lemma}
Given two independent non-negative integer random variables $\Phi$ and $\Delta$. The variable $\Delta$ is distributed as follows
$\Delta$ takes only values $-1$, $0$, $i \cdot H + 1$ with $\Pr{\Delta = -1} = 1$, $\Pr{\Delta = 0} = (1-a-q)$, $\Pr{\Delta = i \cdot H + 1} = \frac{a}{1 - b} \cdot b^i$, where $b \leq \frac{1}{8}$, $a \leq \frac{q}{4H}$, 

If for $\Phi$, we have $\Pr{\Phi \geq k} \leq \left( 1 - \frac{1}{H} \right)^k$ then this bound also holds for $\max\{ \Phi + \Delta, 0 \}$.
\end{lemma}

\begin{proof}
For $k = 0$ the bound trivially holds. So let us consider $k > 0$. Considering all possible values of $\Delta$, we have that $\Pr{ \max\{\Phi + \Delta, 0 \} \geq k}$ is
\[
\Pr{ \Delta = -1, \Phi \geq k + 1} + \Pr{ \Delta = 0, \Phi \geq k} + \sum_{i=0}^\infty \Pr{\Delta = i \cdot H + 1, \Phi \geq k - (i \cdot H + 1)}
\]
Using the definitions and the independence, this is at most
\[
q \cdot \left( 1 - \frac{1}{H} \right)^{k+1} + (1 - q - a) \left( 1 - \frac{1}{H} \right)^k + \sum_{i=0}^\infty \frac{a}{1 - b} b^i \cdot \left( 1 - \frac{1}{H}\right)^{k - (i \cdot H + 1)}
\]
We now apply the fact that $b \leq \nicefrac{1}{8}$ and $H \geq 2$. This yields
\[
\sum_{i = 0}^\infty b^i \cdot \left( 1 - \frac{1}{H} \right)^{- i \cdot H} \leq \sum_{i=0}^\infty (4b)^i \leq 2 \qquad\text{and} \qquad
\frac{1}{1-b} \left( 1 - \frac{1}{H}\right)^{- 1} \leq \frac{5}{2} \enspace.
\]
For this reason, the probability is at most
\[
\left( 1 - \frac{1}{H} \right)^k \left( q \cdot \left( 1 - \frac{1}{H} \right) + (1 - q - a) + 5 a \right)
\leq \left( 1 - \frac{1}{H} \right)^k \left( 1 - \frac{1}{H} q + 4 a \right) \enspace.
\]
As we have $a \leq \frac{q}{4 H}$, this is at most $\left( 1 - \frac{1}{H} \right)^k$.
\end{proof}

Lemmas~\ref{lemma:potentialincrease} and \ref{lemma:potentialdecrease} show that the potential change is stochastically dominated by $\Delta$ in the lemma when setting $H=m^2 J$, $a = (mJ)^{-4}$, $b = (mJ)^{-1}$, and $q = (2 \e m)^{-1}$. Thus the lemma shows that $\Pr{\Phi(t) \geq k} \leq \left( 1 - \frac{1}{m^2 J} \right)^k$ at any time $t$.

\subsection{Packet Latency}
Keeping the insights from the previous section in mind, we can now bound the expected time that a packet spends in the network between the time of injection and reaching its final destination (\emph{latency}). In particular, we show that for each packet with a path length $d$, the expected latency is $O(d \cdot T)$. That is, it takes $O(d)$ time frames.

\begin{theorem}\label{theorem:packetlatency}
The expected latency of a packet of path length $d$ is $O(d \cdot T)$. 
\end{theorem}

For packets that do not fail, this bound is trivial since they take one hop in each time frame. Therefore, it is crucial how much time it takes from the moment a packet fails until its delivery. Fortunately, this can be related to the potential after the time frame of failure.

\begin{observation}\label{observation:packetlatency:potential}
The expected remaining number of time frames a packet spends in the network between failure and reaching its destination is at most $2 \e m \Phi$, where $\Phi$ is the potential after the time frame of failure.
\end{observation}

\begin{proof}
In order to show this claim, we consider the following simplified model that works as an upper bound for the clean-up phases. At the time of failure, all remaining hops of a packet are added to the tail of a FIFO queue. In each time frame, one hop is dequeued with probability $\nicefrac{1}{2 \e m}$. If the queue length is $k$ after adding the hops of a packet, this packet will spend in expectation $2 \e m \cdot \Phi$ time frames in the queue.  

For the actual network, the potential $\Phi$ after the time frame of failure is exactly the number of successful transmission the failed packet has to wait for until it is delivered. Just as in the FIFO queue, in each time frame there is a successful transmission with probability at least $\nicefrac{1}{2 \e m}$. Therefore, the expected number of time frames the packet spends in the network is at most $2 \e m \cdot \Phi$.  
\end{proof}

With this observation, we can proceed to the proof of Theorem~\ref{theorem:packetlatency}.

\begin{proof}[Proof of Theorem~\ref{theorem:packetlatency}]
We consider a single packet from its injection until its delivery. The number of time steps in between is the latency, which is denoted by $L$. The packet is injected during a certain time frame, which we call time frame $0$. It waits at the generator node until the beginning of the next time frame (time frame $1$). Then it crosses one edge in each time frame until it eventually fails or reaches its final destination. By $F$ we denote the number of the time frame in which the packet fails. If the packet reaches its destination without failing, we set $F = \infty$. By $X_i$ we denote the number of time frames that the packet needs from reaching the starting node of the $i$th hop to its final destination. By this definition, we have $L \leq T + X_1 \cdot T$. Set $X_{d+1} = 0$. Furthermore, we let $\Phi_i$ be the potential after the $i$th time frame.

We can write $\Ex{ X_i \mid F \geq i }$, that is the expected time for the hops $i$, \ldots, $d$ given the packet has not failed yet, as follows. Two events can occur: Either the packet does not fail when taking the $i$th hop (i.e. $F \geq i+1$), yielding the remaining number of time frames to be $1 + \Ex{ X_{i+1} \mid F \geq i+1}$. Otherwise we have a failure in exactly that time frame. So formally we get 
\[
\Ex{ X_i \mid F \geq i } = \Pr{ F \geq i + 1 \mid F \geq i } \left( 1 + \Ex{ X_{i+1} \mid F \geq i + 1} \right) + \Pr{ F = i \mid F \geq i } \Ex{ X_i \mid F = i } \enspace.
\]
Multiplying by $\Pr{ F \geq i }$ yields
\[
\Pr{F \geq i} \Ex{ X_i \mid F \geq i } = \Pr{ F \geq i + 1 } \left( 1 + \Ex{ X_{i+1} \mid F \geq i + 1} \right) + \Pr{ F = i } \Ex{X_i \mid F = i } \enspace.
\]
For the latter part, we use Observation~\ref{observation:packetlatency:potential}. In this notation it states $\Ex{X_i \mid F = i} \leq 2 \e m \cdot \Ex{\Phi_i \mid F = i}$. We have already derived bounds on $\Ex{\Phi_i}$ for all $i$. However, $F$ does not necessarily have to be independent since exactly the packet that we are considering might yield a large potential increase.

\begin{claim}
\[
\Pr{ F = i } \Ex{\Phi_i \mid F = i } \leq \frac{2}{mJ}
\]
\end{claim}

\begin{proof}
By definition, we have
\[
\Pr{ F = i } \Ex{\Phi_i \mid F = i } = \sum_{k=1}^\infty \Pr{\Phi_i = k, F = i} \leq (mJ)^3 \cdot \Pr{F = i} + \sum_{k=(mJ)^3+1}^\infty \Pr{\Phi_i \geq k } \enspace.
\]
In the previous section we showed that in each time frame, the probability that packets hail is am most $\frac{1}{(m J)^4}$. Since packet injections are independent, this also yields that for our packet of consideration, the failure probability is each time frame is at most $\frac{1}{(m J)^4}$. Thus, we have $\Pr{ F = i } \leq \frac{1}{(m J)^4}$. That is $(mJ)^3 \cdot \Pr{F = i} \leq \frac{1}{mJ}$. Furthermore, we showed $\Pr{\Phi_i \geq k } \leq \left( 1 - \frac{1}{m^2 J} \right)^{k}$. This yields
\begin{align*}
\sum_{k=(mJ)^3+1}^\infty \Pr{\Phi_i \geq k } & \leq \sum_{k=(mJ)^3+1}^\infty \left( 1 - \frac{1}{m^2 J} \right)^{k} \leq \left( 1 - \frac{1}{m^2 J} \right)^{(mJ)^3} \sum_{k=1}^\infty \left( 1 - \frac{1}{m^2 J} \right)^{k} \\
& \leq \left( 1 - \frac{1}{m^2 J} \right)^{(mJ)^3} m^2 J \leq \frac{1}{mJ} \enspace.
\end{align*}
In combination, this shows the claim.
\end{proof}

Putting in this bound, we now get the following simple linear recursion
\[
\Pr{F \geq i} \Ex{ X_i \mid F \geq i } \leq \Pr{ F \geq i + 1 } \Ex{ X_{i+1} \mid F \geq i + 1} + 3 \enspace.
\]
Solving it, we get
\[
\Ex{ X_1 } = \Pr{F \geq 1} \Ex{ X_1 \mid F \geq 1 } \leq 3 d \enspace, 
\]
which shows that $\Ex{L} \leq \Ex{X_1} T + T \leq 3 d T + T = O(d T)$.
\end{proof}

\section{Dynamic Scheduling Protocol for Adversarial Injection}
In order to transfer the achieved results of the previous section to the adversarial injection model, we adapt an approach by Scheideler and V\"ocking~\cite{Scheideler2000}. The idea is to assign each packet a random delay at the time of injection. Then this packet is kept at the generator node until the delay has elapsed. After this time it is treated as if it was actually injected at this time.

We consider an adversarial injection of rate $\lambda = (1 - \varepsilon) / f(m)$. For each packet a delay value $\delta$ from $0$ to $\delta_{\max} - 1$ is chosen uniformly at random, where $\delta_{\max} = \lceil 2 (D + w) / \varepsilon \rceil$. Like in the stochastic model, it waits until the beginning of the next time frame, but now it spends another $\delta$ time frames waiting. Afterwards it is treated like in the stochastic model with $\lambda' = (1 - \nicefrac{\varepsilon}{2}) / f(m)$.

\begin{theorem}
\label{theorem:adversarial}
The expected queue lengths are bounded at any time. The expected latency of a packet is $O(D \cdot w \cdot T / \varepsilon)$.
\end{theorem}

\begin{proof}
The only point at which the injection model comes into play in the previous section is Claim~\ref{claim:potentialincrease:chernoff}. We have to show that this claim holds for the described protocol in the case of adversarial injection as well. So let us consider a fixed time frame. Let $R(e)$ be again the number of packets that are intended to be transmitted via $e$ in this step. All these packets have been injected at most $D + \delta_{\max}$ time frames ago. Let $\mathcal{P}_j(e)$ be the set of all packets injected in the last $j + \delta_{\max}$ time frames that have $e$ as their $j$th hop. Let $R_{\text{all}}(e) = \sum_{j = 1}^D \lvert \mathcal{P}_j(e) \rvert$. For $p \in \mathcal{P}_j(e)$ let $X_p = 1$ if packet $p$ received a delay such that the $j$th shall take place in the time frame of consideration. With this definition, we have
\[
\Ex{R(e)} = \sum_{j=1}^D \sum_{p \in \mathcal{P}_j(e)} \Ex{X_{p}} = \frac{\sum_{i=1}^D \lvert \mathcal{P}_i (e) \rvert}{\delta_{\max}} = \frac{R_\text{all}(e)}{\delta_{\max}} \enspace.
\]
Thus we get
\[
\Ex{(W \cdot R)(e)} = \sum_{e' \in E} W_{e', e} \frac{R_\text{all}(e)}{\delta_{\max}} \leq \frac{\lVert W \cdot R_\text{all} \rVert_\infty}{\delta_{\max}} \enspace.
\]

All packets in $R_\text{all}$ were injected within at most $\delta_{\max} + D$ time frames, that is at most $(\delta_{\max} + D) \cdot T$ time steps. This interval can be covered by $\lceil (\delta_{\max} + D) \cdot T / w \rceil$ windows of length $w$. By the constraint of the adversary, this yields
\[
\lVert W \cdot R_\text{all} \rVert_\infty \leq \lambda w \lceil (\delta_{\max} + D) \cdot T / w \rceil \leq \lambda (\delta_{\max} + D) \cdot T + \lambda w 
\]

Since we have $\delta_{\max} \geq 2 (D + w) / \varepsilon$, we get
\[
\Ex{(W \cdot R)(e)} \leq \lambda T \frac{\delta_{\max} + D + w}{\delta_{\max}} \leq \lambda T \left( 1 + \frac{\varepsilon}{2} \right) \leq \lambda' T    
\]

Furthermore, for each $e \in E$ the random variable $(W \cdot R)(e)$ is a sum of negatively associated random variables~\cite{Dubhashi1998}. This is, we may apply a Chernoff bound. Combining it with a union bound, we get that for all $\delta > 0$
\[
\Pr{ I \geq (1 + \delta) \lambda' T} \leq \sum_{e \in E} \Pr{(W \cdot R)(e) \geq (1 + \delta) \lambda' T} \leq m \cdot \left( \frac{\e^\delta}{(1 + \delta)^{1+\delta}} \right)^{\lambda' T} \enspace.
\]

That is, the bounds for queue lengths can be transferred. The bound on the latency also holds after having waited the delay. Adding the expected delay yields the claim.
\end{proof}
\section{Application to SINR-based Algorithms}
\label{sec:sinr}
In the SINR model, the network nodes are assumed to be located in a metric space. This allows to model the signal propagation as follows. If some node transmits at power level $p$ then at distance $d$ this signal is received at a strength of $\nicefrac{p}{d^\alpha}$, where $\alpha$ is the so-called path-loss exponent. A transmission can successfully be received if the \emph{signal-to-interference-plus-noise ratio} (SINR) is above some threshold $\beta$. That is, a transmission via a link $\ell = (s, r)$ if for the set $S \subseteq E$ of simultaneous transmissions we have 
\[
\frac{p(\ell)}{d(s, r)^\alpha} \geq \beta \left( \sum_{\substack{\ell' = (s', r') \in S \\ \ell \neq \ell'}} \frac{p(\ell')}{d(s', r)^\alpha} + \noise \right) \enspace.
\]
We choose the impact matrix $W$ depending on whether the transmission powers are fixed for the respective links or they can be chosen for each transmission by the protocol. 

\subsection{Fixed Power Assignments}
Let us first consider the case in which the network links use fixed transmission powers. That is the power value used for a transmission over link $\ell$ is always $p(\ell)$. We define the weight matrix $W$ based on the the relative amount of interference of one link on another one, the so-called \emph{affectance} \cite{Halldorsson2009, Kesselheim2010}. For two links $\ell, \ell' \in E$ it is defined as
\[
a_p(\ell, \ell') = \min \left\{ 1, \beta \frac{p(\ell)}{d(s, r')^\alpha} \Bigg/ \left( \frac{p(\ell')}{d(s', r')^\alpha} - \beta \noise \right) \right\} \enspace.
\]

We achieve the best competitive ratios when dealing with a linear power assignment. That is, $p(\ell)$ is proportional to $d(\ell)^\alpha$ for each link $\ell$ -- and thus the received signal strength is the same for any link. In this case, we set the matrix entries to $W_{\ell, \ell'} = a_p(\ell', \ell)$. With this definition the interference measure $I$ is apart from constant factors the one defined in \cite{Fanghaenel2009}. For this reason, we can use the algorithm from \cite{Fanghaenel2009} that achieves a schedule length of $O(I + \log^2 n)$ whp. Applying the transformation, we get a protocol allowing for injection rates $\Omega(1)$. The lower bound on $I$ in \cite{Fanghaenel2009} states that for each set of transmission requests that can be served in a single step we have $I = O(1)$. Thus the optimally achievable injection rate is $O(1)$ as well. That is, independent of the network size we are only a constant factor worse.

\begin{corollary}
For linear power assignments there is a stable, constant-competitive distributed protocol.
\end{corollary}

Generalizing linear power assignments, we consider power assignments that are (sub-)linear and monotone. That is for two links $\ell, \ell' \in E$ with $d(\ell) \leq d(\ell')$ we have $p(\ell) \leq p(\ell')$ and $p(\ell) / d(\ell)^\alpha \geq p(\ell') / d(\ell')^\alpha$. In this case, we set the matrix $W$ to $W_{\ell, \ell'} = \max\{ a_p(\ell, \ell'), a_p(\ell', \ell) \}$ if $d(\ell) \leq d(\ell')$ and $W_{\ell, \ell'} = 0$ otherwise. We apply the distributed algorithm in \cite{Kesselheim2010}. This algorithm needs for $n$ packet $O(\bar{A} \cdot \log n)$ steps, where $\bar{A}$ denotes the \emph{maximium average affectance} that is define by $\bar{A} = \max_{M \subseteq \mathcal{R}} \avg_{\ell' \in M} \sum_{\ell \in M} a_p(\ell, \ell') = \max_{M \subseteq \R} \frac{1}{\lvert M \rvert} \sum_{\ell' \in M} \sum_{\ell \in M} a_p(\ell, \ell')$. Here, $\mathcal{R}$ denotes the multiset of all transmission requests. We observe that for the interference measure $I$ defined by the matrix $W$, we have $I \geq \bar{A} / 2$. Therefore, we can apply the transformation from Section~\ref{sec:largepacketnumbers} to get a distributed algorithm computing schedules of length $O(I \cdot \log m + \log m \cdot \log^2 m)$ with high probability. This yields a protocol that is stable for all injection rates in $\Omega(1 / \log m)$. Furthermore, the lower bounds on the optimal schedule length in \cite{Kesselheim2010} show that all stable protocols are limited by some injection rate $O(\log m)$.  

\begin{corollary}
For monotone (sub-)linear power assignments there is a stable, $O(\log^2 m)$-competitive distributed protocol. 
\end{corollary}

At this point, one has to remark that in \cite{Halldorsson2011a} an improved analysis of the algorithm in \cite{Kesselheim2010} has been presented. It remains an open problem to fit this analysis into our framework.

\subsection{Powers Chosen by the Algorithm}
There are two approaches to face the setting in which each transmission may use an individual power. On the one hand, one can still define fixed transmission powers for each link in an oblivious fashion, that is, without taking into consideration which transmissions actually take place. Using linear power assignments as described in the previous section, the results in \cite{Fanghaenel2009} yield a $O(\log \Delta \cdot \log m)$-competitive protocols. Here, $\Delta$ is the ratio between the length of the longest and the shortest link. Using square-root power assignments \cite{Fanghaenel2009a, Halldorsson2009a}, we get $O(\log \log \Delta \cdot \log^2 m)$-competitive protocols. Considering fading metrics, that is the setting where $\alpha$ is greater than the doubling dimension, the protocols are $O(\log \Delta)$ respectively $O(\log \log \Delta \cdot \log m)$-competitive. 

We can also exploit the possibility of selecting powers for each transmission individually. For this case only centralized approximation algorithms are known~\cite{Kesselheim2011}. In this case, we set for two links $\ell = (s, r), \ell' = (s', r') \in E$ the weight $W_{\ell, \ell'} = \min \left\{ 1, \frac{d(s, r)^\alpha}{d(s, r')^\alpha} + \frac{d(s, r)^\alpha}{d(s', r)^\alpha} \right\}$ if $d(\ell) \leq d(\ell')$ and $0$ otherwise. The algorithm in \cite{Kesselheim2011} yields schedule lengths of $O(I \cdot \log n)$ with this measure. We have lower bounds of $O(1)$ in fading metrics resp. $O(\log m)$ in general metrics.
\begin{corollary}
For arbitrary transmission powers, there is a stable centralized protocol, that is $O(\log m)$-competitive in fading metrics and $O(\log^2 m)$-competitive in general metrics.
\end{corollary}
This protocol has the drawback of being centralized and for this reason not applicable in practical settings. However, this results shows the problem tractable is in general. In order to construct a distributed protocol, a possible solution could be to spend some time for preprocessing. Even an $O(I \cdot \log m + \mathrm{poly}(m))$ algorithm could be used to get the same competitive ratio.
\section{Further Applications}
\label{sec:furtherapplications}
Defining the matrix $W$ and using the right static algorithm, we can immediately get results for old and new models. For example, for packet routing, setting $W$ to the identity matrix and using the trivial single-hop algorithm, we get stable protocols for all $\lambda<1$. In this section, we demonstrate how to apply our framework in the multiple-access channel or in more involved models by introducing a conflict graph on the network links.

\subsection{Multiple Access Channel}
In order to model the multiple-access channel, we set all entries of the matrix $W$ to $1$. This yields the interference measure to be the number of packets -- a lower bound in this case. We get different results depending on the assumption if stations have individual ids or if all are running the same protocol. More precisely, we get stable protocols for all $\lambda < \nicefrac{1}{\e}$ without station ids. Using ids, we get guarantee stability for all $\lambda < 1$. This matches the best results and the respective lower bounds, see e.g. \cite{Goldberg2000} for details. 

For the symmetric case, i.e. there are no ids, Algorithm~\ref{alg:macstatic} is an acknowledgment-based static algorithm. Using it, we can build stable protocols for all $\lambda<\nicefrac{1}{\e}$.  

\begin{algorithm}
\caption{Static scheduling algorithm for the multiple access channel}
\label{alg:macstatic}
$\xi := \frac{\ln(2 \phi \frac{2 \e (1+\delta)^2}{\delta^2} \ln n)}{\ln( 1 - \frac{1}{\e (1 + \delta)})}$\;
$s := 2 \phi \ln n \frac{2 \e^2 (1 + \delta)^2}{\delta^2}$\;
\For{$i := 1$ to $\xi$}{
Assign each packet a delay independently uniformly at random of at most $\left\lfloor \left( 1 - \frac{1}{\e ( 1 + \delta)} \right)^i n \right\rfloor$\;
Let packets with the same delay be transmitted in the same time step\;
}
\For{$i := 1$ to $s \e ( \phi + 1) \ln n$}{
In each step each packet is transmitted independently with probability $\nicefrac{1}{s}$\;
}
\end{algorithm}

\begin{lemma}
Given constants $\phi \geq 1$ and $\delta > 0$, Algorithm~\ref{alg:macstatic} is a symmetric algorithm for the multiple-access-channel transmitting $n$ packets in $(1+\delta) \e n + O(\phi^2 \log^2 n)$ steps with probability at least $1 - \frac{1}{n^\phi}$
\end{lemma}

\begin{proof}
The total number of time slots used is
\[
\sum_{i=0}^\xi \left( 1 - \frac{1}{\e (1+\delta)} \right)^i \cdot n + \e (\phi + 1) \ln n \cdot \phi \ln n \frac{2 \e^2 (1 + \delta)^2}{\delta^2} \leq (1 + \delta) \e n + O(\phi^2 \log^2 n) \enspace.
\]

Let $X_i$ be the number of packets remaining after the $i$th iteration in the first stage, $X_0 = n$. Lemma~2 from \cite{Goldberg2000} yields the following bound

\[
\Pr{X_i \geq \left( 1 - \frac{1}{\e (1 + \delta)}\right) s \;\Bigg|\; X_{i-1} \leq s} \leq F(s, \delta)\enspace, \quad \text{where } F(s, \delta) = \exp\left(- s \frac{ \delta^2 }{ 2 \e^2 (1 + \delta)^2} \right)
\]

That is, we get
\[
\Pr{X_\xi \geq \left( 1 - \frac{1}{\e (1 + \delta)}\right)^\xi n} \leq \sum_{i = 1}^\xi F\left( \left( 1 - \frac{1}{\e (1 + \delta)}\right)^{i}, \delta\right) \leq \frac{1}{2 n^\phi} \enspace.
\]

For the second stage, we assume that
\[
X_\xi \leq \left( 1 - \frac{1}{\e (1 + \delta)}\right)^\xi n = s\enspace.
\]
This yields that in each step each packet is successfully transmitted with probability at least
\[
\frac{1}{s} \left( 1 - \frac{1}{s} \right)^{s-1} \geq \frac{1}{\e \cdot s} 
\]
This is, the combined probability for a packet not to be successfully transmitted in any of the steps of the second stage is at most
\[
\left( 1 - \frac{1}{\e \cdot s} \right)^{\e s (\phi + 1) \ln n} \leq \frac{1}{2 n^{\phi + 1}}
\]
Taking a union bound yields the claim.
\end{proof}

So, we get a stable protocol for each injection rate $\lambda < \nicefrac{1}{\e}$. This is exactly the same bound as in \cite{Goldberg2000}, unfortunately with a higher packet latency. However, with our transformation the result also hold for adversarial injection.  

\begin{corollary}
There is a symmetric stable protocol for each injection rate $\lambda < \nicefrac{1}{\e}$ on the multiple access channel.
\end{corollary}

For the case that each station has a unique id and stations can distinguish between silence and a successful transmission, a very simple algorithm can still do better.

\begin{lemma}
There is an asymmetric algorithm for the multiple-access channel transmitting $n$ packets in $n + m$ steps.
\end{lemma}

\begin{proof}
The algorithm is straight forward. It was, for example, used as \textsc{Round-Robin-Withholding} by Chlebus et al.~\cite{Chlebus2006} before. The algorithm is deterministic. In each round at most one of the station transmits. Station $0$ starts transmitting in the first round and continues until all requests have been served. Each station $i+1$ listens to the channel while station $i$ transmits. After $i$ has transmitted all packets, in one time slot no transmission is performed at all. This is the signal for station $i+1$ to start transmitting. 
\end{proof}

\begin{corollary}
There is an asymmetric stable protocol for each injection rate $\lambda < 1$ on the multiple-access channel.
\end{corollary}

\subsection{Conflict Graphs}
We can describe further models by a conflict graph. The set of vertices is the set of network links $E$ and (possibly weighted) directed edges indicate if (or to what extent) a transmission on one link is interfered by a transmission on another link. Implementing for example the \emph{node constraint model}, in which each node can only transmit or receive a single packet in each step, we have edges between links that share an endpoint. In this case, we can get constant-competitive since the conflict graph has bounded independence and the algorithm from \cite{Fanghaenel2009} can be adapted.

For the more general case that the conflict graph has inductive independence number $\rho$, we can build $O(\rho \cdot \log m)$-competitive protocols. Conflict graphs with constant $\rho$ for example result from the radio network model in disk graphs, the protocol model or distance-2 matching in disk graphs. The inductive independence number of a graph is defined as follows.

\begin{definition}[\cite{Wan2009,Hoefer2011}]
For a graph $G=(V,E)$, the inductive independence number $\rho$ is the
smallest number such that there is an ordering $\pi$ of the vertices
satisfying: For all $v \in V$ and all independent sets $M \subseteq
V$, we have $\left\lvert M \cap \left\{ u \in V \mid \{u, v\} \in E,
    \pi(u) < \pi(v) \right\} \right\rvert \leq \rho$.
\end{definition}

We define the matrix $W_{e, e'}$ by setting $W_{e, e'} = 1$ if there is an edge $(e, e')$ or $(e', e)$ in the conflict graph and $\pi(e) \leq \pi(e')$. All other matrix entries are set to $0$. This way $I = \max_e \sum_{e' \text{ conflicts with } e, \pi(e') \leq \pi(e)} R(e')$. That is we take the maximum over all edges in the graph and take a summed number of requests at all conflicting edges of smaller index. This yields that no protocol can achieve injection rates greater than $\rho$. 

We consider the following simple distributed algorithm to build an $O(\rho \cdot \log m)$ competitive protocol: In each step, via each link $e$ with probability $1 / 4 I$ each packet is transmitted

\begin{theorem}
The above algorithm needs $O(I \cdot \log n)$ time slots with probability $1 - \nicefrac{1}{n^c}$ for any constant $c$.
\end{theorem}

\begin{proof}
Let $n_t$ be the random variable indicating how many packets still need to be transmitted after the $t$th time slot, $n_0 = n$. Let us consider a fixed time slot $t$. Let $\mathcal{P}$ be the set of remaining packets. Let $X_p$ be the $0$/$1$ random variable indicating if there is a transmission attempt via $e$ in this time slot. We have
\[
\sum_{p \in \mathcal{P}} \Ex{Y_p} \leq \sum_{e \in E} r(e) \Ex{\sum_{e' \in N(e)} X_e} = \sum_{e \in E} r(e) \sum_{e' \in N(e)} \frac{r(e')}{4 I} = \sum_{e \in E} r(e) \cdot \frac{1}{2 I} \sum_{e' \in \Gamma_\pi(e)} r(e') \leq \sum_{e \in E} \frac{r(e)}{2} = \frac{n_t}{2} \enspace.
\]
The probability for half of the packets to collide is at most $\nicefrac{1}{2}$. Each of these packets makes a transmission attempt with probability $\frac{1}{4 I}$. Therefore, we have
\[
\Ex{ n_{t+1} \mid n_t} \leq \left( 1 - \frac{1}{8 I} \right) \cdot n_t \enspace.
\]
This yields the claim.
\end{proof}
\section{Aspects of Distributed Protocols}
In general, it is desirable to design distributed dynamic scheduling protocols. Our transformation requires the nodes to have access to a global clock (in order to build the time frames), and the network size $m$, the injection rate $\lambda$ and (in the adversarial model) to the window size $w$. The other properties depend on the algorithm the protocol was derived from. Particularly, the amount of feedback the protocol needs is identical to the one of the static algorithm. For example, we can start from a static \emph{acknowledgement-based} algorithm, that is the only feedback it gets is if it its own transmission was received. Transforming this algorithm, the dynamic protocol will also be acknowledgement-based. Furthermore, if the algorithm is the same for all nodes, we derive a symmetric protocol.

Fortunately, the required information is available at the time of deployment. So our protocol can be considered distributed if the static algorithm is. However, at this point the natural question arises whether all these assumptions are necessary, particularly the knowledge of a global clock, allowing the construction of common time frames. For the multiple-access channel it can be shown that having local clocks does not weaken the protocols significantly. Even having an acknowledgement-based protocol, local clocks can be synchronized \cite{Goldberg2000}. In our case this is different. We can show that we cannot get $m / 2 \ln m$-competitive without a global clock in the SINR model with uniform transmission powers. This is quite a strong bound because $O(m)$-competitiveness can already be trivially achieved by neglecting geometry aspects and using the multiple-access-channel model.

\begin{theorem}
\label{theorem:localclocks}
There is no stable acknowledgement-based protocol with local clock for the SINR model with uniform transmission powers that is $m / 2 \ln m$-competitive.
\end{theorem}

\begin{proof}
We consider the network given in Figure~\ref{figure:localclocks}. It consists of $m-1$ short links. On each of these links transmissions can be carried out without collisions no matter which other transmissions take place. In contrast, transmissions on the long link can only be successfully carried out when all small links remain silent. Let in each step one packet arrive at each link with probability $\lambda$. Having access to a global clock, we can get stable protocols for all $\lambda < \nicefrac{1}{2}$ by using even time slots for transmissions on the short links and odd time slots for transmissions on the long links. 

\begin{figure}
\begin{center}
\begin{tikzpicture}

        \tikzstyle{vertex}=
        [%
          draw=black,%
          minimum size=1mm,%
          circle,%
          thick%
        ]
 
        \foreach \x in {1,2,...,5,7}
        {
           \node (ns\x) [vertex] at (\x,-.5) {};
           \node (nt\x) [vertex] at (\x,.5) {};
           \path [thick,shorten >=1pt,-stealth'] (ns\x) edge (nt\x);
        }
        
           \node at (6,-.5) {\ldots};
           \node at (6,.5) {\ldots};
 
           \node (nsl) [vertex] at (0,0) {};
           \node (ntl) [vertex] at (8,0) {};
           \path [thick,shorten >=1pt,-stealth'] (nsl) edge (ntl);

\end{tikzpicture}
\end{center}
\caption{The instance considered in the proof of Theorem~\ref{theorem:localclocks}}
\label{figure:localclocks}
\end{figure}
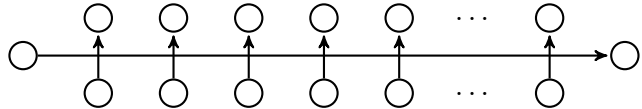

Now we show that for $\lambda \geq \frac{\ln m}{m}$ no acknowledgement-based protocol with only a local clock can be stable. For each small link $i$ and each time slot $t$ let $q_{i, t}$ be the probability that a packet is transmitted via this link in the respective time slot, taking the average over all random packet arrivals. Note that delaying the start of the protocol at a link does not change this behavior. This is due to the fact that it cannot get any feedback form the rest of the network -- as interference is never too high, each transmission attempt is immediately successful.

Since we assume stability there has to be some $k \in \NN$ such that for each interval of length $k$ the expected number of transmissions is at least $(1 - \nicefrac{1}{m}) \lambda$. That is, for each $t_0 \in \NN$, we have $\sum_{t = t_0}^{t_0 + k -1 } q_{i, t} \geq (1 - \nicefrac{1}{m}) \lambda$. Assign each small link a delay from $1, \ldots, k$ independently uniformly at random. This way, for each time slot, the probability that a small link transmits a packet is at least $(1 - \nicefrac{1}{m}) \lambda$. Since the delays were independent, the probability that no small link transmits at all is
\[
\left( 1 - \left(1-\frac{1}{m}\right)\lambda \right)^{m-1} < \frac{\ln m}{\e^{\ln m}} \leq \lambda \enspace.
\]

Now consider the long link. In order to have a successful transmission, none of the small links may be transmitting. As we have just shown, the average number of slots in this happens is at most $\nicefrac{1}{m}$. Thus, even when attempting a transmission in each time slot, we cannot achieve stability for $\lambda \geq \frac{\ln m}{m}$.

%
%
%
%
\end{proof}
\section{Discussion and Open Problems}
In this paper we have shown a general technique to transfer results from static to dynamic packet scheduling in a wireless network. This transformation is independent of the respective interference model. All accesses to the wireless network are performed via a given algorithm for static problems. Improving, adapting or extending this static algorithm suffices to build a new dynamic protocol. This gives a strong motivation for studies of the static scheduling problems.

A possible direction for future work could be considering unreliable networks in the given models. Unreliable communication has been an emerging topic in related fields. For example, an adversarial jammer \cite{Awerbuch2008,Richa2010} and unreliable transmission links \cite{Kuhn2010} in the radio-network model have been considered. Our transformation in principle also allows to be applied on unreliable networks by adapting the respective static algorithm. To name a trivial extension one can consider the case that each transmission is lost with some probability even if interference is small enough. It suffices to consider the effect on the respective static schedule length.

Furthermore, it could be interesting which information is really necessary in which model to design the protocol. We have shown that the global clock is inevitable for our transformation. However, it remains an open question whether knowing the network size, the injection rate and the window size is really necessary to build a protocol.

\bibliographystyle{abbrv}
\bibliography{bibliography}
\clearpage

\end{document}